\documentclass[review]{elsarticle}

\usepackage{lineno,hyperref}
\usepackage{amsmath,amssymb,amsthm, nccmath}
\usepackage{graphicx,color} 
\usepackage{algorithm}
\usepackage[noend]{algpseudocode}
\usepackage{caption}
\usepackage{subcaption}
\usepackage{multicol}
\usepackage[export]{adjustbox}
\usepackage{comment}
\usepackage[T1]{fontenc}
\usepackage[english]{babel}
\usepackage{hyperref}
\usepackage{scrextend}
\usepackage{tikz}
\usetikzlibrary{positioning}
\usepackage{pgfplots}
\usepackage{setspace}

\graphicspath{ {Results/} }

\newtheorem{theorem}{Theorem}

\newtheorem{problem}{Problem}
\newtheorem{fact}{Fact}

\newcommand{\cCW}{{\mathcal {CW}}}
\newcommand{\cW}{{\mathcal {W}}}

\def\beq{\begin{equation}}\def\eeq{\end{equation}}\usepackage[T1]{fontenc}
\usepackage[english]{babel}

\journal{Journal of Information Sciences}

\begin{document}
	
	\begin{frontmatter}
		
		\title{Combinatorial Trace Method for Network Immunization}

		\author[lumsAuthorAddress]{Muhammad Ahmad}\ead{17030056@lums.edu.pk}
		\author[lumsAuthorAddress]{Sarwan Ali}\ead{16030030@lums.edu.pk}
		\author[lumsAuthorAddress]{Juvaria Tariq}\ead{14070004@lums.edu.pk}
		\author[lumsAuthorAddress]{Imdadullah Khan\corref{mycorrespondingauthor} }\ead{imdad.khan@lums.edu.pk}
		\author[ituAuthorAddress]{Mudassir Shabbir}\ead{mudassir.shabbir@itu.edu.pk}
		\author[lumsAuthorAddress]{Arif Zaman}\ead{arifz@lums.edu.pk}
		\cortext[mycorrespondingauthor]{Corresponding author}
		\address[lumsAuthorAddress]{Lahore University of Management Sciences, Pakistan}
		\address[ituAuthorAddress]{Information Technology University, Pakistan}
		
		\begin{abstract}
			Immunizing a subset of nodes in a network - enabling them to identify and withstand the spread of harmful content - is one of the most effective ways to counter the spread of malicious content. It has applications in network security, public health policy, and social media surveillance. Finding a subset of nodes whose immunization results in the least vulnerability of the network is a computationally challenging task. In this work, we establish a relationship between a widely used network vulnerability measure and the combinatorial properties of networks. Using this relationship and graph summarization techniques, we propose an efficient approximation algorithm to find a set of nodes to immunize. We provide theoretical justifications for the proposed solution and analytical bounds on the runtime of our algorithm. We empirically demonstrate on various real-world networks that the performance of our algorithm is an order of magnitude better than the state of the art solution. We also show that in practice the runtime of our algorithm is significantly lower than that of the best-known solution.
		\end{abstract}
		
		
	\end{frontmatter}
	
	
	
	\section{Introduction}
	Graphs or networks are used to model many practical scenarios involving pairwise interactions between entities. The entities could be humans, computers, mobile devices, power components, etc. while interactions can be face-to-face meetings, email and SMS communication and various kind of flows e.g. electric current in a power infrastructure network or fluid in pipelines. Many of the practical networks are very large with millions of nodes and edges.
	
	Every interaction in such large networks can not be monitored and there is a possibility of undesired and potentially harmful communication taking place among entities in networks. Such undesired spread could be intentional or un-intentional entailing various degrees of harms. The unintentional spread of flu-virus, for instance, may be life-threatening and may cause an epidemic. A rumor, on the other hand, may well be originated intentionally and its effect might be limited to a particular segment of a network. An effective way to safeguard a network against the spread of malicious content is to {\em empower} the nodes. The strengthening process may amount to vaccinating people, deploying surveillance systems at junctures and installing anti-virus software on computers depending on the underlying network. The nodes with these added capabilities will be referred to as the {\em immunized} nodes and the malicious content, as the \textit{virus}. Effectively, when a node is immunized, it will neither get contaminated nor will it pass the contaminant to other nodes.
	
	There is a cost associated with immunization, hence it is not feasible to immunize all nodes in large networks. The problem to select a subset of nodes (not exceeding a given budget) for immunization that will maximally hinder the virus spread is called the {\em Network Immunization Problem} and is abstractly formulated in  \cite{chen2016node} as follows: 
	\begin{problem}
		\label{problem:1}
		Given an undirected graph $G = (V,E)$, $|V|=n$ and an integer $k<n$, find a subset $S$ of $k$ nodes such that {\em immunizing} nodes in $S$, renders $G$ the least {\em `vulnerable'} to a virus attack over all choices of $S$.
	\end{problem}
	
	This requires a quantitative measure for the vulnerability of the graph. As in the literature \cite{chen2016node,chakrabarti2008epidemic,ahmad2017spectral,TariqAKS17}, we use the largest eigenvalue of the adjacency matrix of the graph to quantify the graph vulnerability. The objective in Problem \ref{problem:1}, therefore becomes that of immunizing a fixed-sized subset so as the remaining graph has the minimum largest eigenvalue. More precisely,     
	\begin{problem}
		\label{problem:2}
		Given an undirected graph $G = (V,E)$, $|V|=n$ and an integer $k<n$, find a subset $S$ of $k$ nodes such that the largest eigenvalue of the adjacency matrix of $G-S$ (the matrix after removing the rows and columns corresponding to $S$) is minimum over all choices of $S$.
	\end{problem}
	
	A score function was proposed in \cite{ahmad2017spectral,TariqAKS17}, for a subset of nodes based on the number of small length closed walks a node is contained in. The number of fixed length closed walks containing a node co-relates with the node's contribution towards the largest eigenvalue of the adjacency matrix of the graph. However, while the longer walks provide a better approximation, only shorter walks were considered due to time complexity. In this work, we propose a randomized approximation approach to address the time complexity issue, and extend to walks of length $8$, resulting in considerable improvement in accuracy. Formally, the contribution of this work can be summarized as follows: 
	\begin{itemize}\setlength{\itemsep=2pt}		
		\item We extend the score function based on the number of closed walks of length $8$ for sets of nodes that quantify the importance of sets to reduce the graph vulnerability defined in \cite{ahmad2017spectral}. This score function is monotonically non-decreasing and sub-modular that enables employing greedily constructing a set with improved approximation quality 
		\item We derive a closed-form formula to compute the number of walks of length $8$ passing through a node which may be of independent interest. We also give an approximate method that closely estimates the number of walks of length $8$ passing through a node
		\item We evaluate the quality of our solution on several real-world graphs. We show that our approximate method is a close estimate of the exact solution. Results show that our approach maximally reduces the virus spread and the vulnerability (the largest eigenvalue) of the immunized graph.  Moreover, our algorithm is scalable on large graphs and has a lower runtime based on the approximation parameters used. Comparisons also demonstrate that our approach outperforms the state-of-art methods both in terms of quality and runtime
		
	\end{itemize}
	
	The rest of the paper is organized as follows. In Section \ref{section:related_work}, we discuss the related work and give the  background of the problem in Section \ref{section:background}. In Section \ref{section:shieldvalue}, we present our solution along with its analysis. We report experimental results and comparisons with the existing solution in Section \ref{section:experimentalevaluation}.
	

	
	\section{Related Work}\label{section:related_work}
	
	Information spread in networks is widely studied in epidemiology, sociology and information sciences. Researchers are usually interested in estimating the extent to which a contagion will affect the population, predicting the timeline of infection and methods for containing or limiting the effect. The spreading process is studied on a network: agents are represented by nodes and the potential spread of information between a pair of agents is modeled by the presence of an edge between the corresponding pair of nodes. Popular models assume the knowledge of an infection rate $\beta$ (the rate at which an individual/agent accepts content from its neighbors) and a rate of recovery $\delta$ (the rate at which an individual/agent loses content). A relation between spread rate of virus and the largest eigenvalue of adjacency matrix $A$ of the graph,  $\lambda_{max}(A)$, was established in \cite{wang2003epidemic,ganesh2005effect}. In particular, they showed that if $\beta< \delta/\lambda_{max}(A)$, then the infection dies out in sub-linear time with respect to the size of the population following a stochastic model. Similarly, an exponential lower bound on expected die-out time or time for full network recovery (i.e. $\ge e^{cN}$ where $c$ is a constant dependent on the infection rate and $N$ is the size of population)  is also known when $\beta >  \delta/\lambda_{max}(A)$ \cite{van2014exact, van2014upper}. Recent works of \cite{ahn2013global,khanafer2014stability,khanafer2014stability02} established a similar relation of infection and recovery rates with $\lambda_{max}$ for infection spread or die-out while approximating the stochastic model by a deterministic one.

	Various studies have proposed preemptive methods to control virus spread and avoid a potential outbreak of contagion. These methods remove a subset of nodes or edges from the graph, so the remaining graph has the least $\lambda_{max}$. This problem has been shown to be \textsc{Np-Complete} in \cite{chen2016node,van2011decreasing}. An approximation scheme to select nodes for immunization based on eigenvector corresponding to $\lambda_{max}$ of the graph is devised in \cite{chen2016node}.
	
	A combinatorial trace method is adopted in \cite{ ahmad2017spectral,TariqAKS17,Ahmad2016AusDM} to select a subset of nodes whose removal will result in the maximum reduction in the $\lambda_{max}$. The trace of a large power of an adjacency matrix, $A^p$, is closely related to the $\lambda_{max}(A)$ (also known as the spectral radius of the graph) \cite{Abbas2017SRE}. Trace of $A^p$, on the other hand, is just the count of the number of closed walks of length $p$ in the graph. Approximation algorithms are given in \cite{ahmad2017spectral,Ahmad2016AusDM} to select nodes removing which will eliminate the most number of closed walks of length $4$ from the graph. Approximation of the number of closed walks of length $6$ containing a node using a randomly constructed summary of a graph is given in \cite{TariqAKS17}. In this paper, we extend this work by considering walks of length $8$ that leads to improved quality. We note that more sophisticated techniques from graph summarization literature \cite{ lefevre2010grass, riondato2014graph,Riondato2017graph, Beg2018Summarization} could be utilized to improve this work.

	Edge removal techniques are also devised to minimize graph vulnerability. Methods for selecting edges whose removal will reduce $\lambda_{max}$ the most are devised in \cite{kuhlman2013blocking, tong2012gelling}. In \cite{kuhlman2013blocking} virus spread is modeled by the dynamical system and the {\rm transition function} which defines the interaction of a node with its neighbors and state of each node (healthy or infected)   in order to reduce $\lambda_{max}$.
	
	In another line of work, non-preemptive techniques are devised in \cite{zhang2014dava,zhang2014scalable,song2015node} to immunize select nodes after the virus spread has started and the healthy and infected nodes are known. In this setting, methods are evaluated by {\em save ratio (SR)}: the ratio of the number of affected nodes in a graph when $k$ nodes are immunized to the number of infected nodes in case of no immunization.
	
	A reverse engineering technique is used to identify the nodes in a graph where the virus spread is initiated \cite{prakash2012spotting}. A related problem is to decontaminate the graph by deploying {\em cleaning agents} at certain nodes that travel along edges. Monotonicity is assumed in \cite{bienstockseymour1991,flocchini2008,flocchini2007,fraigniaudnisse2008} that a node cleaned by the agent will not get affected again. Non-monotonic strategies are given in \cite{daadaa2016network}.
	
	Some other problems related to graph immunization include the influence maximization \cite{morone2015influence}, the filter placement  \cite{erdos2012filter} and the critical node detection problem (CNDP) \cite{arulselvan2007managing,li2019bi,ventresca2015efficiently}. In the influence maximization problem, the goal is to find a subset of nodes whose \textit{activation} will lead to the maximal spread of information across the graph. The filter placement problem deals with minimizing the multiplicity of information flowing across the network. In CNDP, the goal is to identify nodes whose removal results in maximum graph fragmentation.
	

	
	\section{Preliminaries}\label{section:background}
	
	In this section, we formulate the immunization problem. Given a simple graph $G=(V,E)$, the goal is to select a subset $S$ of $k$ nodes such that removing $S$ from the graph maximally reduces the largest eigenvalue of the remaining graph denoted by $\lambda_{max}(A|_{-S})$. Since $\lambda_{max}$ can be computed in $O(|E|)$, the optimal subset of nodes can be found by iterating through each of the ${n \choose k}$ subsets. The overall runtime of this brute force algorithm is $O({n \choose k}\cdot |E|)$ rendering it computationally infeasible even for moderately large graphs.  
	
	Indeed, it turns out that Problem~\ref{problem:2} is \textsc{NP-Hard}. A reduction from \textit{Minimum Vertex Cover Problem} is as follows: 
	if there exists a set $S$ with $|S|=k$ such that $\lambda_{max}(A|_{-S})=0$, then $S$ is a vertex cover of the graph. It follows from the following implication of famous \textit{Perron-Frobenius theorem}: 
	
	\begin{fact}
		Deleting an edge from a simple connected graph $G$ strictly decreases the largest eigenvalue of the corresponding adjacency matrix \cite{frobenius}. 
	\end{fact}

	Also, if there is a vertex cover $S$ of the graph such that $|S|=k$, then deleting $S$ will result in an empty graph which has eigenvalue zero.
	
	Although Problem 2 is \textsc{NP-Hard}, its objective function is monotone and sub-modular. The greedy algorithm (\textsc{Greedy-1}) guarantees $(1+1/e)$-approximation ($e$ is the base of the natural logarithm) to Problem~\ref{problem:2} by Theorem \ref{NemhauserGreedy}.
	
	
	\begin{theorem}\cite{Nemhauser}\label{NemhauserGreedy} Let $f$ be a non-negative, monotone and submodular  function, $f: 2^{\Omega} \rightarrow \mathbb{R}$. Suppose ${\cal A}$ is an algorithm, that chooses a $k$ elements set $S$ by adding an element $u$ at each step such that  $u= \underset{x\in \Omega\setminus S}{\arg\max}$ $f(S\cup\{x\})$. Then ${\cal A}$ is $(1+1/e)$-approximate algorithm. 
	\end{theorem}
	
	\begin{algorithm}[H]
		\caption{: \textsc{Greedy-1} ($G$,$k$)}
		\label{algo:greedy_eigen}
		\begin{algorithmic}
			\State $S \gets \emptyset$
			\While{$|S|<k$}
			\State $v \gets \underset{x\in V\setminus S}{\arg\min}$ $(\lambda_1(A_{-\{S\cup \{x\} \}}))$
			\State $S \gets S\cup \{v\}$
			\EndWhile
			
			\State \Return $S$ 
		\end{algorithmic}
	\end{algorithm}
	We refer to the achieved benefit after immunizing subset $S$ as {\em eigendrop} and is defined as $\lambda_{max}(A) -  \lambda_{max}(A|_{-S})$. A score, termed as shield-value, is assigned to each subset $S\subset V$, which quantifies the approximated eigendrop achieved after removing $S$. Frequently used symbols in the paper are listed in Table  \ref{table:symbols}. 
	
	\begin{table}[H]
		\centering
		\begin{tabular}{ p{1.55cm}p{9.75cm} }
			
			\hline
			\textbf{Symbol} & \textbf{Definition \&  Description}\\
			\hline
			$A$ & adjacency matrix of the graph $G$\\
			$G|_{-S}$  & subgraph after removing node set S from the graph $G$\\
			$A|_{-S}$  & adjacency matrix of the graph $G|_{-S}$ \\
			
			$\lambda_i(A)$ & $i^{th}$ largest eigen value of matrix $A$ on the basis of magnitude \\
			
			$\lambda_{max}(A)$ & the largest eigen value of matrix $A$ i.e. $ \lambda_{max}(A)  = \lambda_1(A) $ \\
			
			$\Delta \lambda (S)$ & $\lambda_{max}(A) - \lambda_{max}(A|_{-S})$; eigendrop achieved by immunizing node set $S$\\
			
			$A^p$ & $p^{th}$ power of (adjacency) matrix A \\
			
			$\cCW_p(v,G)$ & the set of $p$-length closed walks in $G$ containing $v$\\
			
			$\cCW_p(S,G)$ & the set of $p$-length closed walks in $G$ containing at least one vertex from $S$ \\
			
			$\cW_p(v,G)$ & number of $p$-length closed walks in $G$ containing $v$\\
			
			$\cW_p(S,G)$ & number of $p$-length closed walks in $G$ containing at least one vertex from $S$ \\
			$d_G(v)$ & degree of node $v$ in graph $G$ \\
			
			\hline
			
			\hline
			
		\end{tabular}
		\caption{List of Symbols}
		\label{table:symbols}
	\end{table}
	


	\section{Proposed Shield Value}\label{section:shieldvalue} 
	In this section, we quantify the importance of a subset of nodes for immunization. We first derive a score for each set of size $k$ that closely measures the value of the objective function of Problem \ref{problem:2}. We prove that this score function is monotonically increasing and submodular. Using Theorem \ref{NemhauserGreedy} we can greedily build up the set $S$ by iteratively selecting nodes that are contained in the maximum number of closed walks of length $p$.
	
	Let $A$ be an $n\times n$ matrix; the following two fundamental results from algebraic graph theory \cite{strang1988linear,West2001,Neri19} relate the eigen spectrum and the trace of $A$.
	
	\begin{fact}\label{traceEigen} 
		$$trace(A)=\sum_{i=1}^n A(i,i)=\sum_{i=1}^{n} \lambda_i(A)$$
	\end{fact}
	
	\begin{fact}\label{traceEigenPower} 
		$$trace(A^p)=\sum_{i=1}^{n}\lambda_i(A^p)=\sum_{i=1}^{n}(\lambda_i(A))^p$$
	\end{fact}
	
	From the theory of vector norms \cite{strang1988linear} and Fact \ref{traceEigenPower} we know that 
	\begin{align}\label{traceLambdamax}
	\lim\limits_{\substack{p\rightarrow \infty \\ p \text{ even}}} \left( trace(A^p)\right)^{1/p} \nonumber &= \lim\limits_{\substack{p\rightarrow \infty \\ p \text{ even}}} \left( \sum\limits_{i=1}^n \lambda_i(A)^p\right)^{1/p}  \\
	\nonumber &= \lim\limits_{\substack{p\rightarrow \infty }} \left( \sum\limits_{i=1}^n |\lambda_i(A)|^p\right)^{1/p} = \max_{i}\{\lambda_i(A)\} = \lambda_{max}(A)
	\end{align}
	
	Using the above relation we establish that for the immunization problem, we want to find a subset $S$ of nodes in graph $G$ which, when removed, minimizes $trace((A|_{-S})^p)$. Next, we derive a combinatorial form of this objective function.

	As described in Table \ref{table:symbols}, for a vertex $v\in V(G)$, $\cCW_p(v,G)$ is the set of all closed walks of length $p$ in the graph $G$ containing $v$ at least once and $\cW_p(v,G)=|\cCW_p(v,G)|$. Similarly, $\cCW_p(S,G)$ denotes the set of closed walks of length $p$ in $G$ containing at least one vertex from $S$ and correspondingly $\cW_p(S,G) = |\cCW_p(S,G)|$. We use the following combinatorial definition of $trace$.

	\begin{fact} \label{walkPasTrace}
		\cite{West2001} Given a graph $G = (V,E)$ with adjacency matrix $A$, $$\cW_p(V,G)=trace(A^p)$$
	\end{fact}
	
	From Fact \ref{walkPasTrace} and definition of trace (Fact \ref{traceEigen}), we get that 
	
	\beq\label{eq:Walks_Union} \cW_p(V,G)= \cW_p(V\setminus S, G|_{-S})+ \cW_p(S,G)\eeq
	This is true because any walk in $G$ either contains some vertex in $S$ or it does not contain any vertex in $S$. The former type of walks are counted exactly once in the term $\cW_p(S,G)$, while the first term counts closed walks of the latter type. 
	Equation \eqref{eq:Walks_Union} can be equivalently rewritten as
	
	\begin{align*}
	trace(A^p) = &trace((A|_{-S})^p) + \cW_p(S,G) \\ \implies & trace((A|_{-S})^p) = trace(A^p) - \cW_p(S,G)
	\end{align*} 
	Thus for a fixed graph $G$ (since $trace(A^p)$ is constant) minimizing $trace((A|_{-S})^p)$ is equivalent to maximizing $\cW_p(S,G)$. This implies that the set $S$ with the largest value of $\cW_p(S,G)$ will yield the maximum eigendrop. Intuitively, we need to identify nodes contained in many closed walks of length $p$ (nodes with high $\cW_p(v,G)$). We define the following shield value of a set $S$, that in addition to maximizing $\cW_p(S,G)$, attempts to select those nodes which are far from each other i.e. having $A(u,v)=0$  in order to maximize the number of distinct closed walks going through nodes in a set $S$. 
	
	
	\begin{align}
	\small
	\begin{split}
	\label{walksObjective} score_p(S) &= \gamma \sum_{v\in S} \cW_p(v,G)^2 - \sum_{u,v\in S} \cW_p(v,G)A(u,v) \cW_p(u,G),
	\end{split}
	\end{align}
	where $\gamma$ is a positive constant. Hence Problem 2 can be rephrased as follows.

	\begin{problem}
		\label{problem:3}
		Let $G =(V,E)$ be an undirected graph on $n$ nodes and let $k$ be an integer $k<n$, find a subset of nodes $S\subset V$, with $|S|=k$ such that $score_p(S)$ is the maximum over all $k$-subsets of $V$.
	\end{problem}
	
	For fixed $p$, given $\cW_p(v,G), \forall v \in V$, $score_p(S)$ can be evaluated in time $O(k^2)$ . Selecting a set with maximum $score_p(S)$ takes $O({n\choose k} k^2)$ time which clearly is computationally prohibitive. Furthermore, note that for this we need to have the values of $\cW_p(v,G)$ pre-computed, which is not straight-forward. 
	
	We show that the objective function of Problem \ref{problem:3} is monotone and sub-modular. Given $\cW_p(v,G)$, by Theorem \ref{NemhauserGreedy}, the greedy strategy for building up the set will yield $(1-1/e)$-approximation of the optimal subset.
	
	\begin{theorem} \label{monotonicNonDecreasingTheorem}
		For $p\geq 1$, $score_p(S)$ is monotonically non-decreasing. 
	\end{theorem}
	\begin{proof}
		We prove that for any $X\subset Y \subseteq V$, $score_p(X) \leq score_p(Y)$. Let $E,F\subset V(G)$ and $x\in V(G)$ such that $F=E\cup \{x\}$. Consider
		\begin{equation*}
		\resizebox{1\linewidth}{!}{$\begin{split}
			&score_p(F)-score_p(E)\\
			=& \gamma \sum_{v\in F} \cW_p(v)^2-\sum_{u,v\in F} \cW_p(v)A(u,v) \cW_p(u)-\gamma \sum_{v\in E} \cW_p(v)^2 \\&+\sum_{u,v\in E} \cW_p(v)A(u,v) \cW_p(u) \\
			=& \gamma \cW_p(x)^2-\sum_{v\in E} \cW_p(v)A(x,v) \cW_p(x)	= \cW_p(x)[\gamma \cW_p(x)-\sum_{v\in E} \cW_p(v)A(u,v) ]\geq 0
			\end{split}$}
		\end{equation*}
		Since $\gamma > 0$, for $\gamma\geq k \max_{v\in V(G)} \{\cW_p(v)\}$, the last inequality is satisfied. Hence, $score_p(S)$ function is monotonically non-decreasing.
	\end{proof}
	\noindent
	
	\begin{theorem} \label{submodularityTheorem}
		For $p\geq 1$, $score_p(S)$ is submodular. 
	\end{theorem}
	\begin{proof}
		For any subsets $X,Y$, with $X\subset Y \subseteq V$ and a subset $Z\subset V$ such that $Z\cap Y = \emptyset$, we have $score_p(X\cup Z) -score_p(X)$ is at least as large as $score_p(Y\cup Z) -score_p(Y)$. Let $I,J,K\subset V(G)$ with $I\subset J$. We have
		\begin{equation*}
		\resizebox{1\linewidth}{!}{$\begin{split}
			&score_p(I\cup K)-score_p(I)-score_p(J\cup K)+score_p(J)\\
			=& \Big(\gamma \sum_{v\in I\cup K} \cW_p(v)^2-\sum_{u,v\in I\cup K} \cW_p(v)A(u,v) \cW_p(u) -\gamma \sum_{v\in I} \cW_p(v)^2\\&+\sum_{u,v\in I} \cW_p(v)A(u,v) \cW_p(u)   \Big)-\Big(\gamma \sum_{v\in J\cup K} \cW_p(v)^2-\sum_{u,v\in J\cup K} \cW_p(v)A(u,v) \cW_p(u) \\&-\gamma \sum_{v\in J} \cW_p(v)^2+\sum_{u,v\in J} \cW_p(v)A(u,v) \cW_p(u)   \Big)\\
			=&\Big( \gamma \sum_{v\in K} \cW_p(v)^2-\sum_{u,v\in K} \cW_p(v)A(u,v) \cW_p(u) -2\sum_{u\in K, v\in I} \cW_p(v)A(u,v) \cW_p(u)\Big)\\
			&-\Big(\gamma \sum_{v\in K} \cW_p(v)^2-\sum_{u,v\in K} \cW_p(v)A(u,v) \cW_p(u) -2\sum_{u\in K, v\in J} \cW_p(v)A(u,v) \cW_p(u) \Big) \\
			=& 2\sum_{u\in K, v\in J} \cW_p(v)A(u,v) \cW_p(u) -2\sum_{u\in K, v\in I} \cW_p(v)A(u,v) \cW_p(u) \\ 
			=& 2\sum_{u\in K, v\in J\setminus I} \cW_p(v)A(u,v) \cW_p(u)\geq 0 \end{split}$}\end{equation*}
		\end{proof}
	
	\section{Computing Walks of Length 8}
	The proposed shield value,$score_p(S)$, quantifies the importance of set $S$ based on the number of $p$-length closed walks containing nodes from $S$. Building $S$ requires $\cW_p(v,G)$ for all $v\in V$. A closed-form of $\cW_p(v,G)$ depends on the actual value of $p$. In practice, the value of $p=8$ produces the set $S$ with sufficient quality. We select nodes in a graph based on the number of closed walks of length $8$ (referred to as $8$-walks) for immunization purposes.
	
	\subsection{Justification for p=8}
	Recall that our aim is to find a set $S$ that minimizes $\lambda_{max}(A|_{-S})$. From \eqref{eq:Walks_Union}, we get that for large $p$, $trace(A^p)$ approaches $\lambda_{max}(A)^p$. Hence, we find a set $S$ with minimum $trace(A|^p_{-S})$. We show that in practice $trace(A^8)= \sum_{i=1}^{n} \lambda_i(A^8)$ is sufficiently close to $\lambda_{max}(A^8)$. This is demonstrated by showing that in real world graphs $\dfrac{\lambda_{max}(A^8)}{\sum_{i=1}^{n} \lambda_i(A^8)} = \dfrac{\lambda_{max}(A^8)}{trace(A^8)} $ is close to $1$ specially if there is significant {\it{eigen-gap}} $\big(\lambda_{max}(A) - \lambda_2(A)\big)$. In other words, $\lambda_{max}(A^8)$ is the most dominant term in $trace(A^8)$ and the combined effect of the other terms \big($\lambda_2(A^8) + \cdots + \lambda_n(A^8)$\big) diminishes.
	
	\begin{table}[h!]
		\centering
		\begin{tabular}{ p{4.4cm}p{1.3cm}p{1.3cm}p{1.3cm}p{1.7cm} }
			
			\hline
			\vskip0.1pt
			\textbf{Graph} & \vskip0.1pt
			\textbf{$\vert V \vert$} &\vskip0.1pt
			\textbf{$\lambda_{max}(A)$} & \vskip0.1pt
			\textbf{$\lambda_2(A)$} &\vskip0.1pt
			\textbf{$\dfrac{\lambda_{max}(A^8)}{\sum_{i=1}^{n} \lambda_i(A^8)}$} \vskip 1pt\\
			\hline
			\vskip0.1pt
			EngineeringApplicationofAI & \vskip0.1pt4164 & \vskip0.1pt16 & \vskip0.1pt 13.2 &\vskip0.1pt 0.756\\
			Facebook & 4039 & 162.4 & 125.5 & 0.859 \\
			Email & 1005 & 77.2 & 36.9 & 0.993
			\\
			AICommunication &1203 & 33 & 12.1 & 0.999 \\
			
			\hline
			
			\hline
			
		\end{tabular}
		\caption{Ratio of $\lambda_{max}(A^8)$ to $\sum_{i=1}^{n}\lambda_i(A^8)$ is shown. Note that as relative eigen gap increases, the ratio approaches to $1$. We show the ratio only for moderately large graphs because computing all $n$ eigen values for very large graphs takes very long time.}
		\label{p8justificationExact}
	\end{table}
	
	\subsection{Closed-Form Expression for $\cW_8(v,G)$}
	We derive a closed-form expression for computing $\cW_8(v,G)$. To the best of our knowledge, we are the first one to derive such expression.
	\begin{theorem}\label{walks_closedForm}
		\begin{align*} {}\cW_8(v,G) =&8A^8(v,v)-8A^2(v,v)A^6(v,v) -8A^3(v,v)A^5(v,v)-4(A^4(v,v))^2 \\ 
		& +8A^2(v,v)(A^3(v,v))^2+8(A^2(v,v))^2A^4(v,v) -2(A^2(v,v))^4
		\end{align*}
	\end{theorem}
	
	\begin{proof}
		
		An $8$-walk in $G$ is represented as $W=(a,b,c,d,e,f,g,h,a)$ and the goal is to compute the number of $8$-walks containing a node $v$. Node $v$ can occur at most four times in an $8$-walk and we consider each case of the number of occurrences of $v$ as follows.

		Let $T_{\{l_1,\cdots, l_i\}}, 1\leq i\leq 4$ be the collection of $8$-walks containing $v$ exactly $i$ times. For $W\in T_{\{l_1,\cdots, l_i\}}$, then $W$ starts and ends at $v$ and can be written as concatenation of walks of lengths $l_1, \cdots, l_i$, each starting and ending at $v$. We note that $2\leq l_k\leq 8$, for $1\leq k \leq 4 $, and $\sum_{k=1}^{i} l_k=8$. For example $T_{\{2,3,3\}}$ contains the walks of type $(v,a,v,b,c,v,d,e,v)$ i.e. it is sequence of $(v,a,v), (v,b,c,v)$ and $(v,d,e,v)$.
		
		The rotations of nodes in a walk give different, and sometimes distinct, walks. Given a walk $(a,b,c,d,e,f,g,h,a)$, one vertex left rotation will produce another walk $(b,c,d,e,f,g,h,a,b)$. So recurrent, one vertex, rotations of walks in $T_{\{l_1,\cdots, l_i\}}$ can give up to $8|T_{\{l_1,\cdot, l_i\}}|$ different walks.
		
		We count the walks of each type i.e. walks in $T_{\{2,2,2,2\}}$, $T_{\{2,2,4\}}$, $T_{\{2,3,3\}}$, $T_{\{2,6\}}$, $T_{\{3,5\}}$, $T_{\{4,4\}}, T_{\{8\}}$ and their distinct rotations. In counting there are cases when walks in $T_{\{2,3,3\}}$ are considered and these are different from walks in $T_{\{3,2,3\}}$, but $|T_{\{2,3,3\}}|=|T_{\{3,2,3\}}|$. 
		
		First, we count the number of walks containing $v$ exactly $4$ times. The walk $(v,a,v,b,v,c,v,d,v)$, where $\{a,b,c,d\}  \in N(v)$, is represented as $T_{\{2,2,2,2\}}$ as concatenation of $4$ closed walks of length $2$. The number of such walks is $(A^2(v,v))^4$. In this case, only one vertex rotation is possible which gives $(a,v,b,v,c,v,d,v,a)$ because a second rotation gives the same original walk. Hence, the number of walks containing $v$ exactly $4$ times is $2(A^2(v,v))^4$.
		
		The walks having $v$ exactly $3$ times are contained in $T_{\{2,3,3\}}$ and $T_{\{2,2,4\}}$. The number of walks in $T_{\{2,3,3\}}$ is $A^2(v,v)(A^3(v,v))^2$ and for each walk in this set, $8$ distinct walks are possible after rotations. The total number of walks containing $v$ $3$ times is $8\big[A^2(v,v)(A^3(v,v))^2\big]$.
		
		A walk in $T_{\{2,2,4\}}$ is concatenation of $(v,a,v), (v,b,v),$ $(v,c,d,e,v)$, where $d\neq v$. Number of all walks of form $(v,a,v,b,v,c,d,e,v)$ is at most $8(A^2(v,v))^2A^4(v,v)$ but this number includes walks with $d=v$ as well. To exclude those, we note that when $d=v$, walk is of type $T_{\{2,2,2,2\}}$ which we have already counted in first case. Subtracting the instance when $d=v$ in $(v,a,v,b,v,c,d,e,v)$, we get $|T_{\{2,2,4\}}|=(A^2(v,v))^2A^4(v,v)-(A^2(v,v))^4$. All $8$ vertex rotations of walks in $T_{\{2,2,4\}}$ give distinct walks. The total number of $8$-walks containing $v$ thrice is 
		\begin{align*}
		=&8|T_{\{2,2,4\}}|+8|T_{\{2,3,3\}}|\\
		=&8\left[(A^2(v,v))^2A^4(v,v)-(A^2(v,v))^4\right]+8\left[A^2(v,v)(A^3(v,v))^2\right]
		\end{align*}
		
		Walks containing $v$ exactly twice are represented as  $T_{\{3,5\}}$, $T_{\{2,6\}}$ and $T_{\{4,4\}}$. A walk in $T_{\{3,5\}}$ is of the form $(v,a,b,v,c,d,e,f,v)$ where $d,e\neq v$. The number of walks with $d=v$ and $e=v$ is $|T_{\{3,2,3\}}|$ and $|T_{\{3,3,2\}}|$. So $|T_{\{3,5\}}|=A^3(v,v)A^5(v,v)-2A^2(v,v)(A^3(v,v))^2$. In this case, vertex rotations give $8$ distinct walks.
		
		Walks in $T_{\{2,6\}}$ are of the form $(v,a,v,b,c,d,e,f,v)$ where $c,d,e \neq v$. There are maximum $A^2(v,v)A^6(v,v)$ walks of type $T_{\{2,6\}}$ but these include walks with $c=v$, $d=v$, $e=v$ and $c,e=v$. For $c=v$ and $e=v$, we get walks of types $T_{\{2,2,4\}}$ and $T_{\{2,4,2\}}$ respectively while if $d=v$ then it is a walk of type $T_{\{2,2,2,2\}}$. For $d=v$, we get walk of type $T_{\{2,3,3\}}$.
		
		\begin{align*}
		|T_{\{2,6\}}|=&A^2(v,v)A^6(v,v)-2|T_{\{2,2,4\}}|-|T_{\{2,3,3\}}|-|T_{\{2,2,2,2\}}| \\
		=& A^2(v,v)A^6(v,v)-2(A^2(v,v))^2A^4(v,v)- A^2(v,v)(A^3(v,v))^2 +(A^2(v,v))^4
		\end{align*}
		
		In the case of $T_{\{2,6\}}$, rotations of vertices give $8$ different walks. 
		
		The number of walks of type $T_{\{4,4\}}$ in $(A^4(v,v))^2$ but it also includes $|T_{\{2,4,4\}}|$ and $|T_{\{2,2,2,2\}}|$. Therefore, we get
		
		\begin{align*}
		|T_{\{4,4\}}|&=(A^4(v,v))^2-2|T_{\{2,2,4\}}|-|T_{\{2,2,2,2\}}| \\
		&= (A^4(v,v))^2-2(A^2(v,v))^2A^4(v,v)+(A^2(v,v))^4
		\end{align*}
		
		In this case, only the first $4$ vertex rotations give different walks and $5^{th}$ rotation gives the original walk. The total number of walks containing $v$ exactly twice is 
		
		\begin{align*}
		=&  8|T_{\{3,5\}}|+8|T_{\{2,6\}}|+4|T_{\{4,4\}}| \\
		=& 8\left[A^3(v,v)A^5(v,v)-2A^2(v,v)(A^3(v,v))^2\right] +8[A^2(v,v)A^6(v,v) \\& -2(A^2(v,v))^2A^4(v,v) - A^2(v,v)(A^3(v,v))^2+(A^2(v,v))^4]+4\big[(A^4(v,v))^2 \\& -2(A^2(v,v))^2A^4(v,v)+(A^2(v,v))^4\big]\\
		=&8A^3(v,v)A^5(v,v) +8A^2(v,v)A^6(v,v)+4(A^4(v,v))^2-24A^2(v,v)(A^3(v,v))^2 \\& -24 (A^2(v,v))^2A^4(v,v) +12(A^2(v,v))^4
		\end{align*}
	
		$T_{\{8\}}$ consists of walks containing $v$ only once and are of the form $(v,a,b,c,d,e,f,g,v)$. The number of such walks is $A^8(v,v)$. But this includes walks with some combinations of $b,c,d,e,f$ equal to $v$ as well. Subtracting already counted walks from $T_{\{8\}}$ gives
		
		\begin{align*}
		|T_{\{8\}}|=&A^8(v,v)-2A^2(v,v)A^6(v,v)-2A^3(v,v)A^5(v,v) -(A^4(v,v))^2 \\& +3A^2(v,v))^2A^4(v,v) +3A^2(v,v)(A^3(v,v))^2-(A^2(v,v))^4
		\end{align*} 
		
		In $T_{\{8\}}$, vertex rotations give $8$ distinct walks so the number of walks containing $v$ once is
		\begin{align*}
		8|T_{\{8\}}|=&8A^8(v,v)-16A^2(v,v)A^6(v,v)-16A^3(v,v)A^5(v,v)-8(A^4(v,v))^2 \\& +24A^2(v,v))^2A^4(v,v) +24A^2(v,v)(A^3(v,v))^2-8(A^2(v,v))^4
		\end{align*} 
		Combining all the four cases of occurrence of $v$ in $8$-walk gives 
		\begin{align*}
		\cW_{8}(v,G)=& 2|T_{\{2,2,2,2\}}|+ 8|T_{\{2,2,4\}}|+8|T_{\{2,3,3\}}| + 8|T_{\{3,5\}}|+8|T_{\{2,6\}}|\\&+4|T_{\{4,4\}}| + 8|T_{\{8\}}|\\
		=&8A^8(v,v)-4(A^4(v,v))^2-8A^2(v,v)A^6(v,v) -8A^3(v,v)A^5(v,v) \\& +8A^2(v,v)(A^3(v,v))^2 +8(A^2(v,v))^2A^4(v,v)-2(A^2(v,v))^4
		\end{align*} \vskip-.07in \end{proof}

	\section{Proposed Algorithm}
	In this section, we give our algorithm to compute the number of $8$-walks passing through each vertex and select nodes for immunization. Recall from Theorem \ref{walks_closedForm} that computing number of $8$-walks requires $8^{th}$ power of the adjacency matrix $A$. Let $f(n)$ be the running time for taking $8^{th}$ power of $A$. Computing $\cW_8(v,G)$ for all $v\in V$ using Theorem \ref{walks_closedForm} takes $O(n+ f(n))$ time. Note that while for many real-world graphs $A$ is sparse; this does not necessarily hold for $A^2$ and higher powers of $A$. The above runtime, therefore is prohibitive for real-world graphs, since best-known bounds on $f(n)$ are super-quadratic. 
	
	We propose to approximately compute $\cW_8(v,G)$ from a summary of $G$ \cite{lefevre2010grass,Riondato2017graph,Beg2018Summarization}. Given a graph $G = (V(G),E(G))$ on $n$ nodes, a summary $H$ of $G$, $H=(V(H),E(H))$ is a graph on $t$ nodes with weights on both its nodes and edges. $V(H) = \{V_1,\ldots,V_t\}$ is a partition of $V(G)$, i.e. $V_i\subset V(G)$ for $1\leq i\leq t$, $V_i\cap V_j =\emptyset$ for $i\neq j$ and $\bigcup_{i=1}^t V_i = V(G)$. Each $V_i$ (called supernode) is associated with two integers $n_i=|V_i|$ and $e_i= |\{(u,v)| u,v \in V_i, (u,v) \in E(G) \}|$. Weight of an edge  $(V_i,V_j)\in E(H)$ (called superedge), is  $e_{ij}:$ the number of edges in the bipartite subgraph induced between $V_i$ and $V_j$ i.e. $e_{ij}= |\{(u,v)|u \in V_i, v \in V_j, (u,v) \in E(G) \}|$.
	The original graph $G$ is approximately reconstructed from $H$ as the expected adjacency matrix, $A'_{n \times n}$ with a row and column corresponding to each $u\in V(G)$ given as: $$A'(u,v) = \begin{cases} 0 & \text{ if } u = v\\ \frac{e_i}{{n_i\choose 2}} & \text{ if } u,v\in V_i\\ \frac{e_{ij}}{n_in_j} & \text{ if } u\in V_i, v\in V_j \end{cases}$$
	
	Let $H$ be a summary graph of $G$ on $t$ supernodes and let $C$ be its adjacency matrix. Clearly, $C^p(i,j)$ is the number of walks of length $p$ from nodes in $V_i$ to nodes in $V_j$. We estimate the contributions of $v\in V_i$ to $C^p(i,i)$ by $\alpha_p(v).C^p(i,i)$, where $\alpha_p(v)= \dfrac{d_G(v)^p}{\sum_{u\in V_i} d_G(u)^p}$. Our estimate for $\cW_8(v,G)$ is 
	
	{
		\begin{equation}
		\resizebox{1\textwidth}{!}{$\begin{split}
			\cW'_8(v,G)  =& 8C^8(i,i)\alpha_8(v) -8d_G(v)6C^6(i,i)\alpha_6(v) - 8C^5(i,i)\alpha_5(v)C^3(i,i)\alpha_3(v)  -\\4&\left(C^4(i,i)\alpha_4(v)\right)^2 +8d_G(v)
			\left(C^3(i,i) \alpha_3(v)\right)^2 +8d_G(v)^2C^4(i,i)
			\alpha_4(v) -2d_G(v)^4 \label{approxScoreComputation}
			\end{split}$}
		\end{equation}
	}
	
	This expression is same as that of Theorem \ref{walks_closedForm} except for $p\ge 3$, $A^p(v,v)$ is substituted by $\alpha_p(v).C^p(i,i)$ where $V_i \ni v$. Note that $A^2(v,v)= d_G(v)$. 
	
	We construct a summary $H$ of $G$ by randomly partitioning $V(G)$ into $t$ parts. There are better techniques \cite{lefevre2010grass,Riondato2017graph,Beg2018Summarization} for graph summarization that might result in enhanced estimates.
	
	\subsection{Proposed $\textsc{Walk-8}$ Algorithm}

	We select a subset $S$ that approximately maximizes $score_8(S)$ as given in (\ref{walksObjective}). In Algorithm \ref{algo:greedyalgo}, Line $3$ computes $W$ vector using (\ref{approxScoreComputation}) and $W[i]$ is the estimated number of walks of length $8$ containing vertex $v_i$. In each iteration of Lines $7$-$15$, we greedily extend $S$ by adding a node with the highest score (Line $11$). Line $13$ excludes nodes already selected in $S$ from further consideration. 
	
	\begin{algorithm}[H]
		\caption{: $\textsc{Walk-8}$($A$,$k$,$t$)}
		\label{algo:greedyalgo}
		\begin{algorithmic}[1]
			\State $S \gets \emptyset$
			\State $W_2, Score \gets \Call{zeros}{n} $
			\State $W \gets \Call{EstimateWalks}{A,t}$
			\Comment{compute approx. count of walks using super graph of order $t$ based on Eq. (\ref{approxScoreComputation})}
			\State $\gamma \gets \max_{i} W[i]$
			
			\For{$i=1$ to $n$}
			\State $W_2[i]\gets  \gamma W[i]^2$
			\EndFor
			
			\For{$i=1$ to $k$}
			\State \textbf{u} $\gets A[:, S]*W[S] $
			\For{$j=1$ to $n$}
			\If{$j\notin S$}
			\State $Score[j] \gets W_2[j]-2$\textbf{u}$[j]W[j] $
			\Else
			\State $Score[j] \gets -1$
			\EndIf
			\EndFor
			\State $maxNode \gets \arg \max_{j} Score[j] $
			\State $S \gets S\cup\{maxNode\}$
			\EndFor
			\State \Return $S$
		\end{algorithmic}
	\end{algorithm}
	
	\subsection{Runtime Analysis of $\textsc{Walk-8}$}
	We derive analytical bounds on the runtime of Algorithm \ref{algo:greedyalgo}. Partitioning $G$ into $t$ supernodes takes $O(n)$ time as it can be done with a linear scan on $V(G)$ to put nodes in respective buckets (supernodes). Computing the summary graph (populating the weighted adjacency matrix, $C$) requires traversing the edges $E(G)$ and incrementing the appropriate entry of $C$. This takes a total of $O(|E(G)|)$ time. The powers of $C$ matrix can be computed in $O(t^3)$ time. Thus $\textsc{EstimateWalks}$ function takes $O(n + |E(G)|+t^3)$ time. Line $4$ and the first for loop (Lines $5$-$6$) takes $O(n)$ steps. An iteration of the inner for loop (Lines $9$-$13$) takes $O(n+nk)$ and Line $14$ takes $O(n)$ steps. This shows that the outer loop (Line $7$-$15$) takes $\sum_{i=1}^k O(n+nk)=O(nk^2)$. Therefore, Algorithm \ref{algo:greedyalgo} takes total $O(n + |E(G)|+t^3+nk^2)$ time.

	
	
	\section{Experimental Evaluation}\label{section:experimentalevaluation}
	We present the results of the detailed experimentation of our proposed solution in this section. Experiments are performed on several real-world datasets to analyze the performance of our method and results are compared with $\textsc{NetShield}$\footnote{https://www.dropbox.com/s/aaq5ly4mcxhijmg/Netshieldplus.tar}, the state of art algorithm, to evaluate quality, scalability and efficiency. $\textsc{NetShield}$ computes the score of each node using the eigenvector corresponding to the largest eigenvalue $\lambda_{max}$ of the original graph. $\textsc{Walk-6}$ and $\textsc{Walk-8}$ versions of our algorithm select nodes for immunization based on $6$-walks and $8$-walks respectively passing through each node.

	We evaluate the performance of our algorithm across a range of budgets for the number of nodes to be immunized in the graphs and different counts of supernodes for approximation.
	First, we evaluate the quality of our approximation technique. To show that our approach maximally reduces the spread of the virus across the graph, we give results for the virus spread simulation on graphs immunized by $\textsc{NetShield}$, $\textsc{Walk-6}$ and $\textsc{Walk-8}$. Furthermore, we measure quality in terms of the reduction in $\lambda_{max}$ (\textit{vulnerability}) of the graph after immunizing the set $S$ of selected nodes. We report results using \textit{eigendrop percentage}, which is $\frac{\Delta \lambda (S)}{\lambda_{max}(A)}\times100$. Finally, we give runtime comparisons for the above-mentioned techniques.
	
	 We performed experiments on a standard desktop machine with $3.6$ GHz Intel Core i7-7700 and $8$ GB of main memory. The \textsc{Matlab} code for our algorithm is available  \footnote{\url{https://www.dropbox.com/sh/n7hwjc4imh62pe6/AADCyHG7uMGX6o9xtr1pdH6Qa?dl=0}} for reproducibility and further experimentation. 
	
	\begin{table}[h!]
		\centering
		\begin{tabular}{ p{2.5cm}{l}{l}{l} }
			\hline
			\textbf{Network} & \textbf{Number of Nodes}  & \textbf{Number of Edges} & $\lambda_{max}(A)$\\
			\hline
			
			HEP-TH & 9,877 & 25,998 & 31.03\\
			
			Facebook & 4,039 & 88,234 &162.37\\
			
			Gowalla & 196,591 & 950,327 & 170.94\\
			
			Dblp & 317,080 & 1,049,866 & 115.85\\
			
			Amazon & 334,863 & 925,872 &23.98\\
			
			AA & 418,236 & 2,753,798 & - \\
			
			Youtube & 1,134,890 & 2,987,624 & 210.40\\
			
			Skitter & 1,696,415 & 11,095,298 &670.35\\ 
			
			\hline
			
		\end{tabular}
		\caption{Statistics of Datasets}
		\label{tableOne}
	\end{table}
	\subsection{Datasets} Experiments are performed on real-world graphs of order ranging from a few thousands to a few millions nodes. All graphs are undirected and unweighted.
	HEP-TH\footnote{https://snap.stanford.edu/ \label{snap_repository}} is a collaboration network of High Energy Physics - Theory category extracted from the e-print arXiv. A node in the network represents an author and an edge between two authors shows collaboration between them. Facebook\footref{snap_repository} graph shows the friendship network among users in which people are represented as nodes and relationships among two users are shown as edges. 
	
	To test our algorithm on large networks we use five different real-world graphs. Gowalla\footref{snap_repository} dataset shows friendship relations in a location-based social network. Amazon\footref{snap_repository} is a co-purchasing graph of products where each node is a product and there is an edge between two nodes if the products are purchased by a user in a single basket. Dblp\footref{snap_repository} is a co-authorship network in which two authors are connected if they have co-authored at least one publication. Youtube\footref{snap_repository} graph shows the friendship network of users in the Youtube social network. Skitter\footref{snap_repository} is an internet topology network where nodes correspond to autonomous systems and communication between them constitutes edges.
	
	\begin{table}[h!]
		\centering
		\begin{tabular}{ p{5.1cm}p{1.8cm}p{1.8cm}{l} }
			
			\hline
			\textbf{Network} & \textbf{Number of Nodes}  & \textbf{Number of Edges} & $\lambda_{max}(A)$\\
			\hline
			Applied Mathematics and \newline Computing (AMC)
			& 18,371 & 24,224 & 10.99\\
			
			Decision Support Systems (DSS) & 4,926 & 14,660 & 12.0\\
			
			Ecological Informatics (EI) & 1,990 & 4,913 & 16.68\\
			
			Communication ACM & 11,476 & 16,687 & 32.90\\
			
			\hline
			
		\end{tabular}
		\caption{Statistics of AA subgraphs}
		\label{table2}
	\end{table}
	
The dataset AA\footnote{http://dblp.uni-trier.de/xml/} is a co-authorship network extracted from DBLP archive data. We select $4$ different smaller co-authorship subgraphs each corresponding to manuscripts in a distinct journal. Node count goes up to a few thousands and edge count goes up to a few ten thousands for extracted subgraphs. Details of the subgraphs of AA data set are provided in Table \ref{table2}.

	\subsection{Approximation Quality of $\textsc{Walk-8}$}
	In order to evaluate the goodness of our approximate method, we compare it with the exact solution as described in Theorem \ref{walks_closedForm}. The exact number of closed walks of length $8$ can be computed using the original adjacency matrix $A$ as given in Theorem \ref{walks_closedForm} instead of using a summary graph. We analyze the quality of our approximation method by comparing the eigendrop percentages achieved using the exact and approximate method. We report comparison results of the exact solution with the summary graphs of order $\{100,500,1000\}$.
	
	\pgfplotsset{title style={at={(0.85,0.75)}}}
	\pgfplotsset{every x tick label/.append style={font=\tiny}}
	\pgfplotsset{every y tick label/.append style={font=\tiny}}
	\pgfplotsset{compat=1.5}
	\noindent
	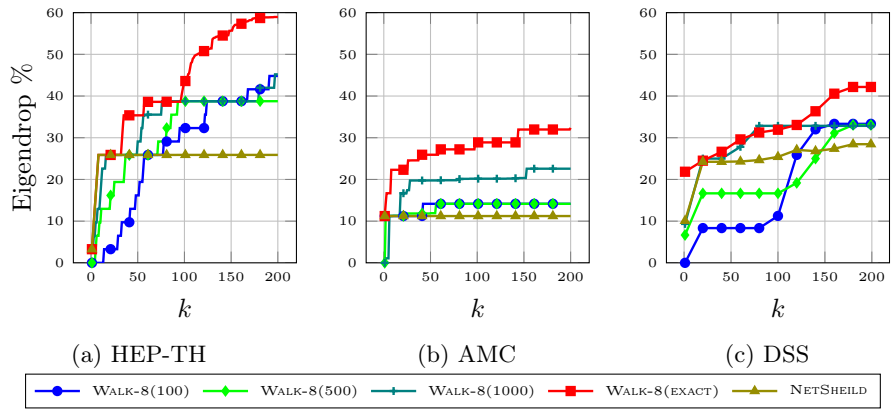
\begin{figure}[h!]
		\centering
		\begin{subfigure}[b]{0.3\columnwidth}
			\begin{tikzpicture}   
			\begin{axis}[ 
			xlabel={$k$},
			ylabel={Eigendrop \%}, ylabel shift=-3pt,
			xtick={0, 50, 100, 150, 200},
			ytick={0,10,20,30,40,50,60},
			height=1.35\columnwidth, width=1.25\columnwidth, grid=major,
			ymin=0,
			ymax=60,
			legend columns=-1,
			legend style={
				column sep=1ex,
			},
			legend entries={\textsc{NetSheild},\textsc{Walk$_8(100)$},\textsc{Walk$_8(500)$}, \textsc{Walk$_8(1000)$},\textsc{Walk$_8$(exact)}},
			mark repeat={20},
			legend to name=commonlegend
			]
			
			\addplot[blue,mark size=1.5pt, mark= *, solid,line width=0.9] table[x={x},y={cw8_100}]{hep_th_approximation_comparison.csv};
			
			\addplot[green,mark size=1.5pt, mark = diamond*, solid ,line width=0.9] table[x={x},y={cw8_500}]{hep_th_approximation_comparison.csv};
			
			\addplot[teal,mark size=1.5pt, mark= +, solid,line width=0.9] table[x={x},y={cw8_1000}]{hep_th_approximation_comparison.csv};
			
			\addplot[red,mark size=1.5pt, mark = square*, solid ,line width=0.9] table[x={x},y={cw_8Exact}]{hep_th_approximation_comparison.csv};
			
			\addplot[olive,mark size=1.5pt, mark =triangle*, line width=0.9] table[x={x},y={NetSheild}]{hep_th_approximation_comparison.csv};
			\end{axis}
			\end{tikzpicture}%
			\caption{HEP-TH}
		\end{subfigure}
		\hspace*{.04\textwidth}
		\begin{subfigure}[b]{0.3\columnwidth}
			\begin{tikzpicture}   
			\begin{axis}[ 
			xlabel={$k$},
			xtick={0, 50, 100, 150, 200},
			ytick={0,10,20,30,40,50,60},
			height=1.35\columnwidth, width=1.25\columnwidth, grid=major,
			ymin=0,
			ymax=60,
			legend columns=-1,
			legend style={
				column sep=1ex,
			},
			mark repeat={20},
			legend to name=commonlegend
			]
			
			\addplot[blue,mark size=1.5pt, mark= *, solid,line width=0.9] table[x={x},y={cw8_100}]{AppliedMathematicsAndComputing_approximation_comparison.csv};
			
			\addplot[green,mark size=1.5pt, mark = diamond*, solid ,line width=0.9] table[x={x},y={cw8_500}]{AppliedMathematicsAndComputing_approximation_comparison.csv};
			
			\addplot[teal,mark size=1.5pt, mark= +, solid,line width=0.9] table[x={x},y={cw8_1000}]{AppliedMathematicsAndComputing_approximation_comparison.csv};
			
			\addplot[red,mark size=1.5pt, mark = square*, solid ,line width=0.9] table[x={x},y={cw_8Exact}]{AppliedMathematicsAndComputing_approximation_comparison.csv};
			
			\addplot[olive,mark size=1.5pt, mark =triangle*, line width=0.9] table[x={x},y={NetSheild}]{AppliedMathematicsAndComputing_approximation_comparison.csv};
			\end{axis}
			\end{tikzpicture}
			\caption{AMC}
		\end{subfigure}
		\hspace*{.01\textwidth}
		\begin{subfigure}[b]{0.3\columnwidth}
			\begin{tikzpicture}   
			\begin{axis}[ 
			xlabel={$k$},
			xtick={0, 50, 100, 150, 200},
			ytick={0,10,20,30,40,50,60},
			height=1.35\columnwidth, width=1.25\columnwidth, grid=major,
			ymin=0,
			ymax=60,
			legend columns=-1,
			legend style={
				column sep=1ex,
			},
			legend entries={\textsc{Walk-$8(100)$},\textsc{Walk-$8(500)$}, \textsc{Walk-$8(1000)$},\textsc{Walk-$8$(exact)},\textsc{NetSheild}},
			legend to name=commonlegend
			]
			
			\addplot[blue,mark size=1.5pt, mark= *, solid,line width=0.9] table[x={x},y={cw8_100}]{DecisionSupportSystems_clusteringEffect_1.csv};
			
			\addplot[green,mark size=1.5pt, mark = diamond*, solid ,line width=0.9] table[x={x},y={cw8_500}]{DecisionSupportSystems_clusteringEffect_1.csv};
			
			\addplot[teal,mark size=1.5pt, mark= +, solid,line width=0.9] table[x={x},y={cw8_1000}]{DecisionSupportSystems_clusteringEffect_1.csv};
			
			\addplot[red,mark size=1.5pt, mark = square*, solid ,line width=0.9] table[x={x},y={cw_8Exact}]{DecisionSupportSystems_clusteringEffect_1.csv};
			
			\addplot[olive,mark size=1.5pt, mark =triangle*, line width=0.9] table[x={x},y={NetSheild}]{DecisionSupportSystems_clusteringEffect_1.csv};
			\end{axis}
			\end{tikzpicture}
			\caption{DSS}
		\end{subfigure}
		\\\tiny\ref{commonlegend}
		\caption{The effect of the order of summary graph on the quality of the approximation. Eigen-drop percentages using different numbers of supernodes have been reported ($\textsc{Walk-8}(t)$, where $t$ is the number of supernodes).  It is clear that as $t$ increases, the quality of approximation tends to match with that of the exact solution.} \label{fig:exactvsapproxcomparison}
	\end{figure}

	It is clear from Figure \ref{fig:exactvsapproxcomparison} that the performance of our approximate method improves with the increase in the number of supernodes in the summary graph. As the order of the summary graph increases, the achieved benefit tends to match with that of the exact solution. Note that we compute the exact number of walks for small graphs having the order of a few thousands only as it is computationally infeasible to compute the exact solution for large graphs.

	\pgfplotsset{title style={at={(0.85,0.75)}}}
	\pgfplotsset{every x tick label/.append style={font=\tiny}}
	\pgfplotsset{every y tick label/.append style={font=\tiny}}
	\pgfplotsset{compat=1.5}
	\noindent
	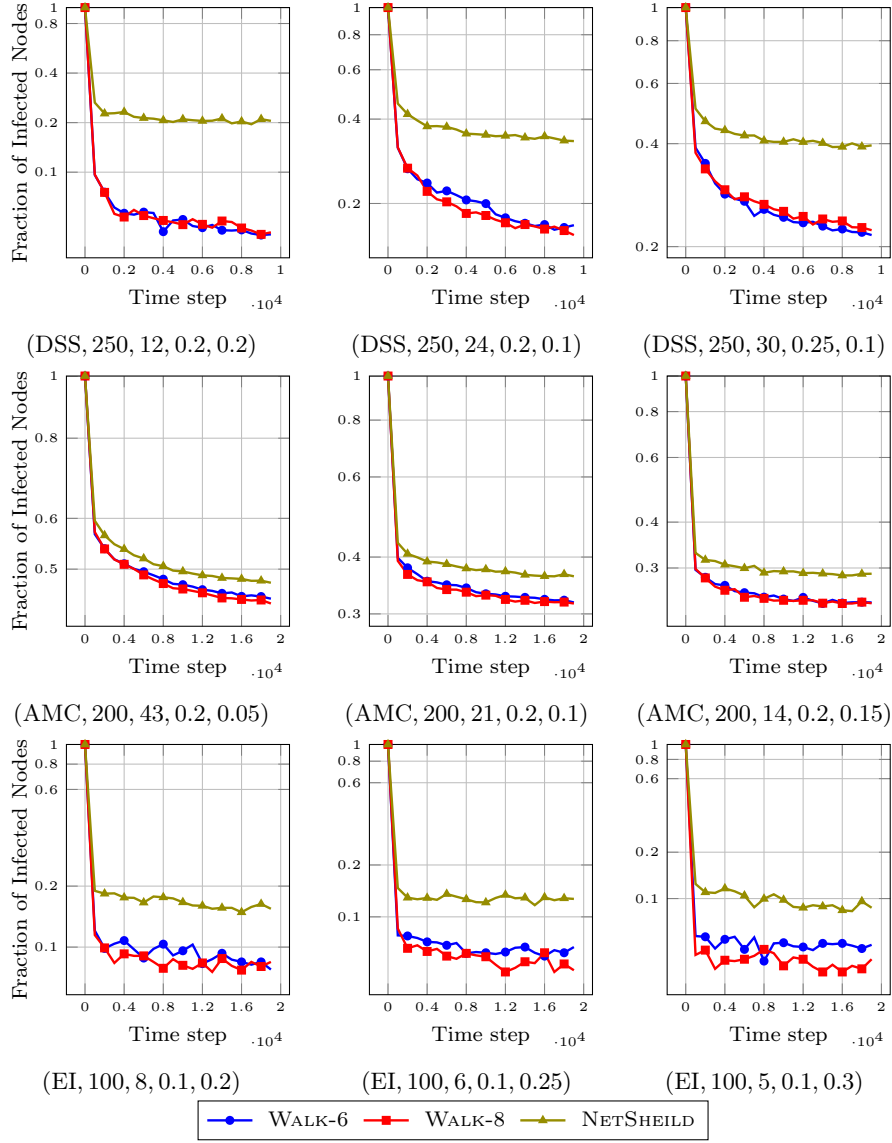
\begin{figure}[h!]
		\centering\footnotesize
		
		\begin{subfigure}[b]{0.3\columnwidth}
			\begin{tikzpicture}   
			\begin{axis}[ 
			xlabel={Time step},
			ylabel={Fraction of Infected Nodes}, ylabel shift=-3pt,
			xtick={0, 2000, 4000, 6000, 8000,10000},
			ytick={0,0.1,0.2,0.4,0.8,1},
			height=1.35\columnwidth, width=1.25\columnwidth, grid=major,
			ymin=0,
			ymax=1,
			log ticks with fixed point,
			ymode=log,
			legend columns=-1,
			legend style={
				column sep=1ex,
			},
			mark repeat={2},
			legend to name=commonlegend
			]
			\addplot[blue,mark size=1.25pt, mark = *, line width=0.9] table[x={x},y={cw6}]{DecisionSupportSystems_250_0.2_0.2_10000_truncated.csv};
			\addplot[red,mark size=1.25pt, mark= square*, solid,line width=0.9] table[x={x},y={cw8}]{DecisionSupportSystems_250_0.2_0.2_10000_truncated.csv};
			\addplot[olive,mark size=1.25pt, mark = triangle*, solid ,line width=0.9] table[x={x},y={ns}]{DecisionSupportSystems_250_0.2_0.2_10000_truncated.csv};
			\end{axis}
			\end{tikzpicture}%
			\tiny{\subcaption*{$(\text{DSS},250,12,0.2,0.2)$}}  
		\end{subfigure}
		\hspace*{.04\textwidth}
		\begin{subfigure}[b]{0.3\columnwidth}
			\begin{tikzpicture}  
			\begin{axis}[
			xlabel={Time step},
			ylabel shift=-3pt,
			xtick={0, 2000, 4000, 6000, 8000,10000},
			ytick={0,0.2,0.4,0.6,0.8,1},
			height=1.35\columnwidth, width=1.25\columnwidth, grid=major,
			ymin=0, ymax=1,
			ymode = log,
			log ticks with fixed point,
			legend columns=-1,
			legend style={
				column sep=1ex,
			},
			mark repeat={2},
			legend to name=commonlegend
			]
			\addplot[blue,mark size=1.25pt, mark = *, line width=0.9] table[x={x},y={cw6}]{DecisionSupportSystems_250_0.2_0.1_10000_truncated.csv};
			
			\addplot[red,mark size=1.25pt, mark= square*, every mark/.append style={solid}, line width=0.9] table[x={x},y={cw8}]{DecisionSupportSystems_250_0.2_0.1_10000_truncated.csv};
			
			\addplot[olive,mark size=1.25pt, mark = triangle*, every mark/.append style={solid},line width=0.9] table[x={x},y={ns}]{DecisionSupportSystems_250_0.2_0.1_10000_truncated.csv};
			\end{axis}
			\end{tikzpicture}
			\footnotesize\subcaption*{{$(\text{DSS},250,24,0.2,0.1)$}}
		\end{subfigure}
		\hspace*{.01\textwidth}
		\begin{subfigure}[b]{0.3\columnwidth}
			\begin{tikzpicture}  
			
			\begin{axis}[
			xlabel={Time step},
			height=1.35\columnwidth, width=1.25\columnwidth, grid=major,
			ymin=0, ymax=1,
			ymode = log,
			log ticks with fixed point,
			legend columns=-1,
			legend style={
				column sep=1ex,
			},
			xtick={0, 2000, 4000, 6000, 8000,10000},
			ytick={0,0.2,0.4,0.6,0.8,1},
			legend entries={\textsc{Walk$_6$}, \textsc{Walk$_8$}, \textsc{NetSheild}},
			mark repeat={2},
			legend to name=commonlegend
			]
			\addplot[blue,mark size=1.25pt, mark = *, line width = 0.9] table[x={x},y={cw6}]{DecisionSupportSystems_250_0.25_0.1_10000_truncated.csv};
			
			\addplot[red,mark size=1.25pt, mark= square*, solid , line width = 0.9] table[x={x},y={cw8}]{DecisionSupportSystems_250_0.25_0.1_10000_truncated.csv};
			
			\addplot[olive,mark size=1.25pt, mark = triangle*, solid, line width = 0.9] table[x={x},y={ns}]{DecisionSupportSystems_250_0.25_0.1_10000_truncated.csv};
			\end{axis}
			
			\end{tikzpicture}
			\subcaption*{$(\text{DSS},250,30,0.25,0.1)$}
		\end{subfigure}
		
		\begin{subfigure}[b]{0.3\columnwidth}
			\begin{tikzpicture}   
			\begin{axis}[ 
			xlabel={Time step},
			ylabel={Fraction of Infected Nodes}, ylabel shift=-3pt,
			xtick={0, 4000, 8000, 12000, 16000,20000},
			ytick={0,0.4,0.5,0.6,0.8,1},
			height=1.35\columnwidth, width=1.25\columnwidth, grid=major,
			ymin=0,
			ymax=1,
			log ticks with fixed point,
			ymode=log,
			legend columns=-1,
			legend style={
				column sep=1ex,
			},
			mark repeat={2},
			legend to name=commonlegend
			]
			\addplot[blue,mark size=1.25pt, mark = *, line width=0.9] table[x={x},y={cw6}]{AppliedMathematicsAndComputing_200_0.2_0.05_20000_truncated.csv};
			\addplot[red,mark size=1.25pt, mark= square*, solid,line width=0.9] table[x={x},y={cw8}]{AppliedMathematicsAndComputing_200_0.2_0.05_20000_truncated.csv};
			\addplot[olive,mark size=1.25pt, mark = triangle*, solid ,line width=0.9] table[x={x},y={ns}]{AppliedMathematicsAndComputing_200_0.2_0.05_20000_truncated.csv};
			\end{axis}
			\end{tikzpicture}%
			\tiny{\subcaption*{$(\text{AMC},200,43,0.2,0.05)$}} 
		\end{subfigure}
		\hspace*{.04\textwidth}
		\begin{subfigure}[b]{0.3\columnwidth}
			\begin{tikzpicture}  
			\begin{axis}[
			xlabel={Time step},
			ylabel shift=-3pt,
			xtick={0, 4000, 8000, 12000, 16000,20000},
			ytick={0,0.3,0.4,0.6,0.8,1},
			height=1.35\columnwidth, width=1.25\columnwidth, grid=major,
			ymin=0, ymax=1,
			ymode = log,
			log ticks with fixed point,
			legend columns=-1,
			legend style={
				column sep=1ex,
			},
			mark repeat={2},
			legend to name=commonlegend
			]
			\addplot[blue,mark size=1.25pt, mark = *, line width=0.9] table[x={x},y={cw6}]{AppliedMathematicsAndComputing_200_0.2_0.1_20000_truncated.csv};
			
			\addplot[red,mark size=1.25pt, mark= square*, every mark/.append style={solid}, line width=0.9] table[x={x},y={cw8}]{AppliedMathematicsAndComputing_200_0.2_0.1_20000_truncated.csv};
			
			\addplot[olive,mark size=1.25pt, mark = triangle*, every mark/.append style={solid},line width=0.9] table[x={x},y={ns}]{AppliedMathematicsAndComputing_200_0.2_0.1_20000_truncated.csv};
			\end{axis}
			\end{tikzpicture}
			\footnotesize\subcaption*{{$(\text{AMC},200,21,0.2,0.1)$}}
		\end{subfigure}
		\hspace*{.01\textwidth}
		\begin{subfigure}[b]{0.3\columnwidth}
			\begin{tikzpicture}  
			
			\begin{axis}[
			xlabel={Time step},
			height=1.35\columnwidth, width=1.25\columnwidth, grid=major,
			ymin=0, ymax=1,
			ymode = log,
			log ticks with fixed point,
			legend columns=-1,
			legend style={
				column sep=1ex,
			},
			xtick={0, 4000, 8000, 12000, 16000,20000},
			ytick={0,0.3,0.4,0.6,0.8,1},
			legend entries={\textsc{Walk$_6$}, \textsc{Walk$_8$}, \textsc{NetSheild}},
			mark repeat={2},
			legend to name=commonlegend
			]
			\addplot[blue,mark size=1.25pt, mark = *, line width = 0.9] table[x={x},y={cw6}]{AppliedMathematicsAndComputing_200_0.2_0.15_20000_truncated.csv};
			
			\addplot[red,mark size=1.25pt, mark= square*, solid , line width = 0.9] table[x={x},y={cw8}]{AppliedMathematicsAndComputing_200_0.2_0.15_20000_truncated.csv};
			
			\addplot[olive,mark size=1.25pt, mark = triangle*, solid, line width = 0.9] table[x={x},y={ns}]{AppliedMathematicsAndComputing_200_0.2_0.15_20000_truncated.csv};
			\end{axis}
			
			\end{tikzpicture}
			\subcaption*{$(\text{AMC},200,14,0.2,0.15)$}
		\end{subfigure}
		
		\begin{subfigure}[b]{0.3\columnwidth}
			\begin{tikzpicture}   
			\begin{axis}[ 
			xlabel={Time step},
			ylabel={Fraction of Infected Nodes}, ylabel shift=-3pt,
			xtick={0, 4000, 8000, 12000, 16000,20000},
			ytick={0,0.1,0.2,0.6,0.8,1},
			height=1.35\columnwidth, width=1.25\columnwidth, grid=major,
			ymin=0,
			ymax=1,
			log ticks with fixed point,
			ymode=log,
			legend columns=-1,
			legend style={
				column sep=1ex,
			},
			mark repeat={2},
			legend to name=commonlegend
			]
			\addplot[blue,mark size=1.25pt, mark = *, line width=0.9] table[x={x},y={cw6}]{EcologicalInformatics_100_0.1_0.2_20000_truncated.csv};
			\addplot[red,mark size=1.25pt, mark= square*, solid,line width=0.9] table[x={x},y={cw8}]{EcologicalInformatics_100_0.1_0.2_20000_truncated.csv};
			\addplot[olive,mark size=1.25pt, mark = triangle*, solid ,line width=0.9] table[x={x},y={ns}]{EcologicalInformatics_100_0.1_0.2_20000_truncated.csv};
			\end{axis}
			\end{tikzpicture}%
			\tiny{\subcaption*{$(\text{EI},100,8,0.1,0.2)$}} 
		\end{subfigure}
		\hspace*{.04\textwidth}
		\begin{subfigure}[b]{0.3\columnwidth}
			\begin{tikzpicture}  
			\begin{axis}[
			xlabel={Time step},
			ylabel shift=-3pt,
			xtick={0, 4000, 8000, 12000, 16000,20000},
			ytick={0,0.1,0.2,0.6,0.8,1},
			height=1.35\columnwidth, width=1.25\columnwidth, grid=major,
			ymin=0, ymax=1,
			ymode = log,
			log ticks with fixed point,
			legend columns=-1,
			legend style={
				column sep=1ex,
			},
			mark repeat={2},
			legend to name=commonlegend
			]
			\addplot[blue,mark size=1.25pt, mark = *, line width=0.9] table[x={x},y={cw6}]{EcologicalInformatics_100_0.1_0.25_20000_truncated.csv};
			
			\addplot[red,mark size=1.25pt, mark= square*, every mark/.append style={solid}, line width=0.9] table[x={x},y={cw8}]{EcologicalInformatics_100_0.1_0.25_20000_truncated.csv};
			
			\addplot[olive,mark size=1.25pt, mark = triangle*, every mark/.append style={solid},line width=0.9] table[x={x},y={ns}]{EcologicalInformatics_100_0.1_0.25_20000_truncated.csv};
			\end{axis}
			\end{tikzpicture}
			\footnotesize\subcaption*{{$(\text{EI},100,6,0.1,0.25)$}}
		\end{subfigure}
		\hspace*{.01\textwidth}
		\begin{subfigure}[b]{0.3\columnwidth}
			\begin{tikzpicture}  
			
			\begin{axis}[
			xlabel={Time step},
			height=1.35\columnwidth, width=1.25\columnwidth, grid=major,
			ymin=0, ymax=1,
			ymode = log,
			log ticks with fixed point,
			legend columns=-1,
			legend style={
				column sep=1ex,
			},
			xtick={0, 4000, 8000, 12000, 16000,20000},
			ytick={0,0.1,0.2,0.6,0.8,1},
			legend entries={\textsc{Walk-$6$}, \textsc{Walk-$8$}, \textsc{NetSheild}},
			mark repeat={2},
			legend to name=commonlegend
			]
			\addplot[blue,mark size=1.25pt, mark = *, line width = 0.9] table[x={x},y={cw6}]{EcologicalInformatics_100_0.1_0.3_20000_truncated.csv};
			
			\addplot[red,mark size=1.25pt, mark= square*, solid , line width = 0.9] table[x={x},y={cw8}]{EcologicalInformatics_100_0.1_0.3_20000_truncated.csv};
			
			\addplot[olive,mark size=1.25pt, mark = triangle*, solid, line width = 0.9] table[x={x},y={ns}]{EcologicalInformatics_100_0.1_0.3_20000_truncated.csv};
			\end{axis}
			
			\end{tikzpicture}
			\subcaption*{$(\text{EI},100,5,0.1,0.3)$}
		\end{subfigure}
		
		\ref{commonlegend}
		\caption{Virus propagation simulation for varying virus strength $s$ on the immunized version of graphs. Caption of each plot represents (graph name, number of immunized nodes $k$, $s$, infection rate $\beta$, recovery rate $\delta$). Initially, all the nodes in the graphs were contaminated and the plots show the fraction of infected nodes ($y$-axis logged scale) as the time proceeds. }
		\label{fig:virusPropagation}
	\end{figure}
	\subsection{Virus Spread Simulation}
	Another criterion used for quality evaluation is to estimate the spread of virus propagation in the immunized version of the graph. We use \textsc{SIR} virus propagation model to observe the spread of contagion after immunizing a small subset ($\sim$ 5 \%) of nodes in a graph. Let $s = \lambda_{max} \times \beta/\delta$ be the virus strength (larger value of $s$ corresponds to more strength of virus while the virus gradually dies out if $s\le 1$), where $\beta$ and $\delta$ denote the infection and recovery rate respectively. In our experimentation, we immunize $k$ nodes in a graph and infect all the nodes in the immunized version of the graph. We then observe the spread of the virus under different virus strengths with varying values of $\beta$ and $\delta$. Results in Figure \ref{fig:virusPropagation} show that the graphs immunized by our approach have less number of infected nodes as compared to \textsc{NetSheild}. We report the average of $3$ runs of experiments to mitigate the effect of randomness.

	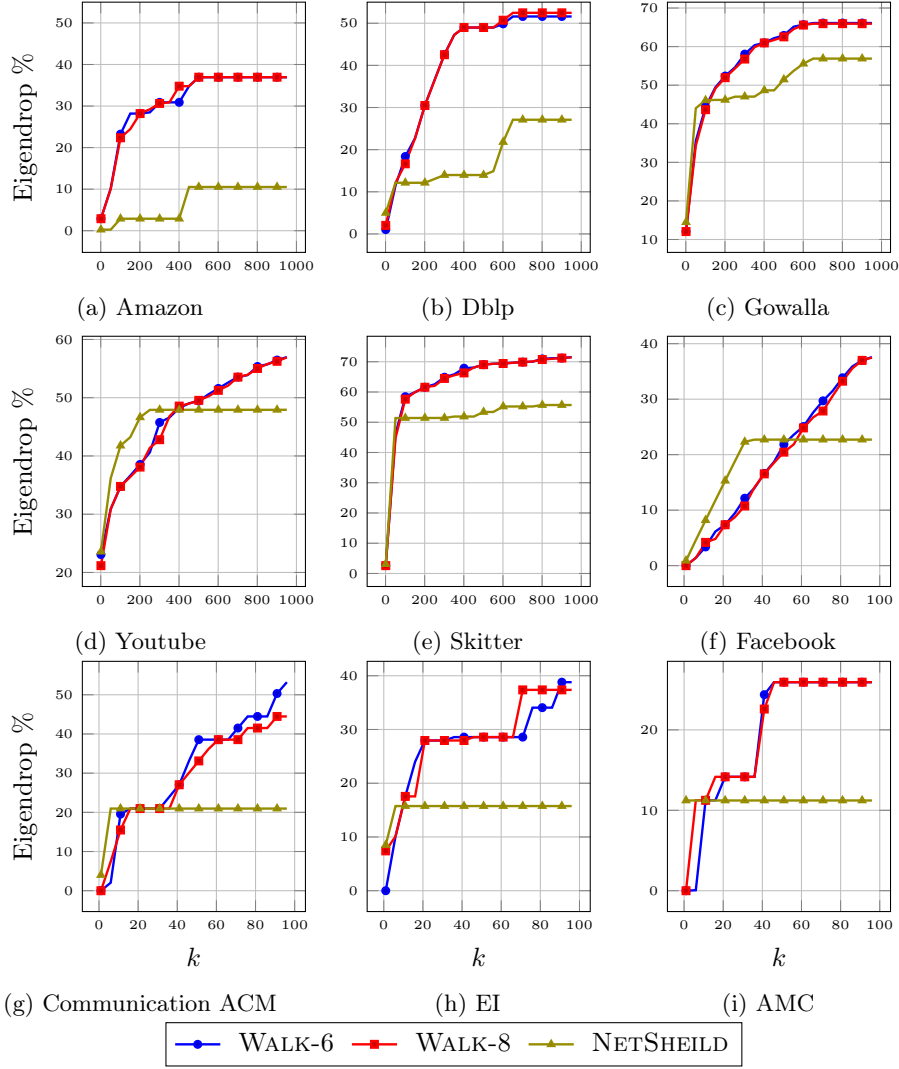
\begin{figure}[h!]  
		\centering 
		\begin{subfigure}[b]{0.3\columnwidth}
			\begin{tikzpicture}   
			\begin{axis}[title={},
			ylabel={Eigendrop \%},
			ymax = 55,
			xtick={0, 200, 400, 600, 800,1000},
			ytick={0,10,20,30,40,50},
			height=1.35\columnwidth, width=1.25\columnwidth, grid=major,
			legend columns=-1,
			x tick label style={/pgf/number format/.cd,%
				set thousands separator={}},
			legend style={
				column sep=1ex,
			},
			mark repeat={2},
			legend to name=commonlegend_eigen_drop_percentage
			]
			\addplot+[blue,mark size=1.25pt, mark = *, solid ,line width=0.9] table[x={x},y={cw6}]{amazon_eigenDropPercentage.csv};
			\addplot+[red,mark size=1.25pt, mark = square*, solid ,line width=0.9] table[x={x},y={cw8}]{amazon_eigenDropPercentage.csv};
			\addplot+[olive,mark size=1.25pt, mark = triangle*, solid ,line width=0.9] table[x={x},y={ns}]{amazon_eigenDropPercentage.csv};
			\end{axis} 
			\end{tikzpicture}%
			\caption{Amazon} \label{fig:test1}  
		\end{subfigure}
		\hspace*{.04\textwidth}
		\begin{subfigure}[b]{0.3\columnwidth}
			\begin{tikzpicture}  
			\begin{axis}[title={},
			ylabel={},
			ymax = 55,
			xtick={0, 200, 400, 600, 800,1000},
			ytick={0,10,20,30,40,50},
			height=1.35\columnwidth, width=1.25\columnwidth, grid=major,
			legend columns=-1,
			x tick label style={/pgf/number format/.cd,%
				set thousands separator={}},
			legend style={
				column sep=1ex,
			},
			mark repeat={2},
			legend to name=commonlegend_eigen_drop_percentage
			]
			\addplot+[blue,mark size=1.25pt, mark = *, solid ,line width=0.9] table[x={x},y={cw6}]{dblp_eigenDropPercentage.csv};
			\addplot+[red,mark size=1.25pt, mark = square*, solid ,line width=0.9] table[x={x},y={cw8}]{dblp_eigenDropPercentage.csv};
			\addplot+[olive,mark size=1.25pt, mark = triangle*, solid ,line width=0.9] table[x={x},y={ns}]{dblp_eigenDropPercentage.csv};
			\end{axis} 
			\end{tikzpicture}
			\caption{Dblp} \label{fig:test2}  
		\end{subfigure}
		\hspace*{.01\textwidth}
		\begin{subfigure}[b]{0.3\columnwidth}
			\begin{tikzpicture}  
			\begin{axis}[title={},
			ylabel={},
			xtick={0, 200, 400, 600, 800,1000},
			ytick={0,10,20,30,40,50,60,70},
			height=1.35\columnwidth, width=1.25\columnwidth, grid=major,
			legend columns=-1,
			x tick label style={/pgf/number format/.cd,%
				set thousands separator={}},
			legend style={
				column sep=1ex,
			},
			mark repeat={2},
			legend to name=commonlegend_eigen_drop_percentage
			]
			\addplot+[blue,mark size=1.25pt, mark = *, solid ,line width=0.9] table[x={x},y={cw6}]{gowalla_eigenDropPercentage.csv};
			\addplot+[red,mark size=1.25pt, mark = square*, solid ,line width=0.9] table[x={x},y={cw8}]{gowalla_eigenDropPercentage.csv};
			\addplot+[olive,mark size=1.25pt, mark = triangle*, solid ,line width=0.9] table[x={x},y={ns}]{gowalla_eigenDropPercentage.csv};
			\end{axis} 
			\end{tikzpicture}
			\caption{Gowalla} \label{fig:test3}  
		\end{subfigure}

		
		\begin{subfigure}[b]{0.3\columnwidth}
			\begin{tikzpicture}   
			\begin{axis}[title={},
			ylabel={Eigendrop \%}, 
			xtick={0, 200, 400, 600, 800,1000},
			ytick={0,10,20,30,40,50,60},
			height=1.35\columnwidth, width=1.25\columnwidth, grid=major,
			legend columns=-1,
			x tick label style={/pgf/number format/.cd,%
				set thousands separator={}},
			legend style={
				column sep=1ex,
			},
			mark repeat={2},
			legend to name=commonlegend_eigen_drop_percentage
			]
			\addplot+[blue,mark size=1.25pt, mark = *, solid ,line width=0.9] table[x={x},y={cw6}]{youtube_eigenDropPercentage.csv};
			\addplot+[red,mark size=1.25pt, mark = square*, solid ,line width=0.9] table[x={x},y={cw8}]{youtube_eigenDropPercentage.csv};
			\addplot+[olive,mark size=1.25pt, mark = triangle*, solid ,line width=0.9] table[x={x},y={ns}]{youtube_eigenDropPercentage.csv};
			\end{axis} 
			\end{tikzpicture}%
			\caption{Youtube} \label{fig:test4}  
		\end{subfigure}
		\hspace*{.04\textwidth}
		\begin{subfigure}[b]{0.3\columnwidth}
			\begin{tikzpicture}  
			\begin{axis}[title={},
			ylabel={},
			xtick={0, 200, 400, 600, 800,1000},
			ytick={0,10,20,30,40,50,60,70},
			height=1.35\columnwidth, width=1.25\columnwidth, grid=major,
			legend columns=-1,
			x tick label style={/pgf/number format/.cd,%
				set thousands separator={}},
			legend style={
				column sep=1ex,
			},
			mark repeat={2},
			legend to name=commonlegend_eigen_drop_percentage
			]
			\addplot+[blue,mark size=1.25pt, mark = *, solid ,line width=0.9] table[x={x},y={cw6}]{skitter_eigenDropPercentage.csv};
			\addplot+[red,mark size=1.25pt, mark = square*, solid ,line width=0.9] table[x={x},y={cw8}]{skitter_eigenDropPercentage.csv};
			\addplot+[olive,mark size=1.25pt, mark = triangle*, solid ,line width=0.9] table[x={x},y={ns}]{skitter_eigenDropPercentage.csv};
			\end{axis} 
			\end{tikzpicture}
			\caption{Skitter} \label{fig:test5}  
		\end{subfigure}
		\hspace*{.01\textwidth}
		\begin{subfigure}[b]{0.3\columnwidth}
			\begin{tikzpicture}  
			\begin{axis}[title={},
			ylabel={},
			xtick={0, 20, 40, 60, 80,100},
			ytick={0,10,20,30,40,50,60,70},
			height=1.35\columnwidth, width=1.25\columnwidth, grid=major,
			legend columns=-1,
			x tick label style={/pgf/number format/.cd,%
				set thousands separator={}},
			legend style={
				column sep=1ex,
			},
			mark repeat={2},
			legend to name=commonlegend_eigen_drop_percentage
			]
			\addplot+[blue,mark size=1.25pt, mark = *, solid ,line width=0.9] table[x={x},y={cw6}]{facebook_eigenDropPercentage.csv};
			\addplot+[red,mark size=1.25pt, mark = square*, solid ,line width=0.9] table[x={x},y={cw8}]{facebook_eigenDropPercentage.csv};
			\addplot+[olive,mark size=1.25pt, mark = triangle*, solid ,line width=0.9] table[x={x},y={ns}]{facebook_eigenDropPercentage.csv};
			\end{axis} 
			\end{tikzpicture}
			\caption{Facebook} \label{fig:test6}  
		\end{subfigure}

		
		
		\begin{subfigure}[b]{0.3\columnwidth}
			\begin{tikzpicture}   
			\begin{axis}[title={},
			ylabel={Eigendrop \%},
			xlabel={$k$},
			xtick={0, 20, 40, 60, 80,100},
			ytick={0,10,20,30,40,50,60},
			height=1.35\columnwidth, width=1.25\columnwidth, grid=major,
			legend columns=-1,
			x tick label style={/pgf/number format/.cd,%
				set thousands separator={}},
			legend style={
				column sep=1ex,
			},
			mark repeat={2},
			legend to name=commonlegend_eigen_drop_percentage
			]
			\addplot+[blue,mark size=1.25pt, mark = *, solid ,line width=0.9] table[x={x},y={cw6}]{CommunicationACM_eigenDropPercentage.csv};
			\addplot+[red,mark size=1.25pt, mark = square*, solid ,line width=0.9] table[x={x},y={cw8}]{CommunicationACM_eigenDropPercentage.csv};
			\addplot+[olive,mark size=1.25pt, mark = triangle*, solid ,line width=0.9] table[x={x},y={ns}]{CommunicationACM_eigenDropPercentage.csv};
			\end{axis} 
			\end{tikzpicture}%
			\caption{Communication ACM} \label{fig:test7}  
		\end{subfigure}
		\hspace*{.04\textwidth}
		\begin{subfigure}[b]{0.3\columnwidth}
			\begin{tikzpicture}  
			\begin{axis}[title={},
			ylabel={},
			xlabel={$k$},
			xtick={0, 20, 40, 60, 80,100},
			ytick={0,10,20,30,40},
			height=1.35\columnwidth, width=1.25\columnwidth, grid=major,
			legend columns=-1,
			x tick label style={/pgf/number format/.cd,%
				set thousands separator={}},
			legend style={
				column sep=1ex,
			},
			mark repeat={2},
			legend to name=commonlegend_eigen_drop_percentage
			]
			\addplot+[blue,mark size=1.25pt, mark = *, solid ,line width=0.9] table[x={x},y={cw6}]{ecologicalInformatics_eigenDropPercentage.csv};
			\addplot+[red,mark size=1.25pt, mark = square*, solid ,line width=0.9] table[x={x},y={cw8}]{ecologicalInformatics_eigenDropPercentage.csv};
			\addplot+[olive,mark size=1.25pt, mark = triangle*, solid ,line width=0.9] table[x={x},y={ns}]{ecologicalInformatics_eigenDropPercentage.csv};
			\end{axis} 
			\end{tikzpicture}
			\caption{EI} \label{fig:test8}  
		\end{subfigure}
		\hspace*{.01\textwidth}
		\begin{subfigure}[b]{0.3\columnwidth}
			\begin{tikzpicture}  
			\begin{axis}[title={},
			ylabel={},
			xlabel={$k$},
			xtick={0, 20, 40, 60, 80,100},
			ytick={0,10,20,30},
			height=1.35\columnwidth, width=1.25\columnwidth, grid=major,
			legend columns=-1,
			x tick label style={/pgf/number format/.cd,%
				set thousands separator={}},
			legend style={
				column sep=1ex,
			},
			legend entries={\textsc{Walk-$6$}, \textsc{Walk-$8$},\textsc{NetSheild}},
			mark repeat={2},
			legend to name=commonlegend_eigen_drop_percentage
			]
			\addplot+[blue,mark size=1.25pt, mark = *, solid ,line width=0.9] table[x={x},y={cw6}]{AppliedMathematicsAndComputing_eigenDropPercentage.csv};
			\addplot+[red,mark size=1.25pt, mark = square*, solid ,line width=0.9] table[x={x},y={cw8}]{AppliedMathematicsAndComputing_eigenDropPercentage.csv};
			\addplot+[olive,mark size=1.25pt, mark = triangle*, solid ,line width=0.9] table[x={x},y={ns}]{AppliedMathematicsAndComputing_eigenDropPercentage.csv};
			\end{axis} 
			\end{tikzpicture}
			\caption{AMC} \label{fig:test9}  
		\end{subfigure}
		
		\ref{commonlegend_eigen_drop_percentage}
		\caption{Comparison of $\textsc{NetSheild}$, $\textsc{Walk-6}$ and $\textsc{Walk-8}$ in terms of eigendrop percentages ($y$-axis) against budget $k$, number of nodes immunized, ($x$-axis). $\textsc{Walk-6}$ and $\textsc{Walk-8}$ achieve significantly higher eigendrop for increasing $k$. Results in (a)-(e) are computed using $1000$ supernodes while in (f)-(i) experiments are performed using summary graph of order $= 500$. The range for $k$ is chosen keeping in view the number of nodes in the host graphs.} \label{fig:eigenDropPercentages}
	\end{figure}  
	
	\subsection{EigenDrop Percentage Comparison}
	We compare the quality of approximate versions of our algorithms with \textsc{NetSheild} in terms of eigendrop and results are shown in Figure \ref{fig:eigenDropPercentages}. For smaller graphs and subgraphs of AA which consist of a few thousand nodes, a budget of up to $100$ nodes is used and for large graphs with more than $100,000$ nodes, we immunize up to $1000$ nodes. We have used summary graphs with different orders $(100,500,1000)$ to perform experiments. Time complexity increases as the number of supernodes increases but we observe that there is a proportionately minor improvement in the quality of solution for increasing order of graph after a certain threshold is reached. For smaller graphs, we report results for supernode count of $500$ and for large graphs, the number of supernodes is set to $1000$.
	
	We observe that the immunizing quality of our algorithm clearly outperforms \textsc{NetSheild} in terms of eigendrop. The improvement in quality of solution is particularly evident on large graphs Gowalla Figure~\ref{fig:test3}, Youtube Figure~\ref{fig:test4}, and Skitter Figure~\ref{fig:test5}. For reasonably large budget,  \textsc{Walk-8} outperforms both \textsc{NetSheild} and \textsc{Walk-6}. Experiments also reveal that \textsc{NetSheild} performs better than our approach for very small values of budget $k$ but as the count of nodes to be immunized increases, its effectiveness degrades.

	\subsection{Run Time Comparison} We also present comparable computational cost while achieving much better quality as one of the merits of our algorithm as discussed in the theoretical time complexity in Section~\ref{section:shieldvalue}. Comparison of runtimes of \textsc{NetSheild}, \textsc{Walk-6} and \textsc{Walk-8} is provided in Figure~\ref{fig:immunizationTime}. Results show that the runtime of our algorithm matches with that of $\textsc{NetSheild}$. The results are reported with $1000$ supernodes ($t$) in summary graphs. 
	
	\pgfplotsset{title style={at={(0.8,0.05)}}}
	\pgfplotsset{compat=1.5}
	\pgfplotsset{every x tick label/.append style={font=\tiny}}
	\noindent
	\begin{figure}[H]
		\centering\footnotesize
		\begin{tikzpicture}
		\begin{axis}[title={Amazon},
		xlabel={$k$},
		ylabel={Time Taken(s)}, ylabel shift=-3pt,
		ymin=0,ymax=5,
		xtick={0, 200, 400, 600, 800,1000},
		ytick={0,1,2,3,4,5},
		height=0.45\columnwidth, width=0.41\columnwidth, grid=major,
		legend columns=-1,
		x tick label style={/pgf/number format/.cd,%
			set thousands separator={}},
		legend style={
			column sep=1ex,
		},
		mark repeat={2},
		legend to name=commonlegend_time
		]
		\addplot+[blue,mark size=1.25pt, mark = *, solid ,line width=0.9] table[x={x},y={walk6_time}]{amazon_time_tikz.csv};
		\addplot+[red,mark size=1.25pt, mark = square*, solid ,line width=0.9] table[x={x},y={walk8_time}]{amazon_time_tikz.csv};
		\addplot+[olive,mark size=1.25pt, mark = triangle*, solid ,line width=0.9] table[x={x},y={ns_time}]{amazon_time_tikz.csv};
		\end{axis}
		\end{tikzpicture}
		\begin{tikzpicture}
		\begin{axis}[title={Gowalla},
		xlabel={$k$},
		ymin=0,ymax=5,
		height=0.45\columnwidth, width=0.41\columnwidth, grid=major,
		legend columns=-1,
		x tick label style={/pgf/number format/.cd,%
			set thousands separator={}},
		legend style={
			column sep=1ex,
		},
		xtick={0, 200, 400, 600, 800,1000},
		ytick={0,1,2,3,4,5},
		mark repeat={2},
		legend to name=commonlegend_time
		]
		\addplot+[blue,mark size=1.25pt, mark = *, solid ,line width=0.9] table[x={x},y={walk6_time}]{gowalla_time_tikz.csv};
		\addplot+[red,mark size=1.25pt, mark = square*, solid ,line width=0.9] table[x={x},y={walk8_time}]{gowalla_time_tikz.csv};
		\addplot+[olive,mark size=1.25pt, mark = triangle*, solid ,line width=0.9] table[x={x},y={ns_time}]{gowalla_time_tikz.csv};
		\end{axis}
		\end{tikzpicture}
		\begin{tikzpicture}
		\begin{axis}[title={Youtube},
		xlabel={$k$},
		height=0.45\columnwidth, width=0.41\columnwidth, grid=major,
		legend columns=-1,
		x tick label style={/pgf/number format/.cd,%
			set thousands separator={}},
		legend style={
			column sep=1ex,
		},
		xtick={0, 200, 400, 600, 800,1000},
		ytick={0,4,8,12,16,20,22},
		legend entries={\textsc{Walk-$6$}, \textsc{Walk-$8$}, \textsc{NetSheild}},
		mark repeat={2},
		legend to name=commonlegend_time
		]
		\addplot+[blue,mark size=1.25pt, mark = *, solid ,line width=0.9] table[x={x},y={walk6_time}]{youtube_time_tikz.csv};
		\addplot+[red,mark size=1.25pt, mark = square*, solid ,line width=0.9] table[x={x},y={walk8_time}]{youtube_time_tikz.csv};
		\addplot+[olive,mark size=1.25pt, mark = triangle*, solid ,line width=0.9] table[x={x},y={ns_time}]{youtube_time_tikz.csv};
		\end{axis}
		\end{tikzpicture}	
		\ref{commonlegend_time}
		\caption{Comparison of time taken (in seconds) to immunize graphs using $\textsc{NetSheild}$, $\textsc{Walk-6}$ and $\textsc{Walk-8}$ approach against the number of nodes to be immunized $(k)$. Results are reported on summaries with $1000$ supernodes.}
		\label{fig:immunizationTime}
	\end{figure}
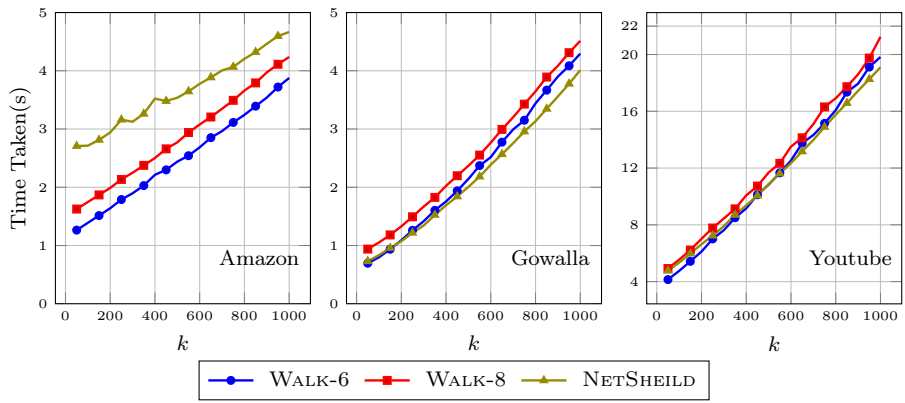
	
	Recall that runtime of our algorithm is $O(n + |E(G)|+t^3+nk^2)$, where the first three terms comprise runtime of constructing a summary of order $t$ and computing the $\cW_{8}(v,G)$ for all $v \in V(G)$, while the last term ($nk^2$) is the runtime to select the best $k$ nodes (\textsc{NetSheild} also requires $O(nk^2)$ for this task). Hence runtime of our algorithm depends quadratically only on $k$, which generally is a small constant. We note that our runtime is superior to that of \textsc{NetSheild} in the sense that in relatively less time we achieve significantly more eigendrop even for a small value of $t$ (see Figure \ref{fig:exactvsapproxcomparison}). 
	
	
	
	\section{Conclusion}
	
	In this work, we address the problem of finding a small subset of nodes in a network whose immunization results in a significant reduction in network vulnerability towards the spread of undesirable content. We explored the relationships between spectral and graph-theoretic properties of networks and exploit these relationships to design an efficient algorithm to find crucial nodes in the network. We select a subset of nodes for immunization based on the number of closed walks of length $8$. With the use of easily computable local graph properties and approximation techniques, the running time of our technique is linear in the size of the graph. Thus, our method is scalable and can be applied to large graphs.
	Experiments on large real-world networks suggest that our algorithm provides better results than previously employed methods and is significantly faster in terms of time complexity. The approximation quality comparison shows that our method is a close approximation of the exact solution. Experimental results for various quality measures like virus spread simulation, reduction in network vulnerability and the run time comparison show that our method performs better than the state of the art solution.
	
	Potential extensions of this work include i) utilizing specialized graph summarization methods, this will further reduce computational cost as well as improve immunization performance of the solution ii) extending this work to incorporate dynamic graphs. Dynamic graphs evolve with time and edges are added/removed iii) exploring non-preemptive graph immunization approaches, where the immunization process starts after the virus attack and the information of infected nodes is available. 
	
	
	\bibliography{combinatorial_immunization_IS_Revised}
	
\end{document}